%% file: softring.tex
\documentclass[a4paper, 11pt]{amsart}
\usepackage{longtable, stackrel, multicol, fancybox, tcolorbox}
\usepackage{bm,amsfonts,amsmath,amssymb,mathrsfs,bbm,amsthm}
\usepackage{graphicx, color,hyperref,eurosym, eucal,palatino, marginnote, stackrel}

\definecolor{darkblue}{rgb}{0.2,0.2,0.6}
\definecolor{darkblue2}{rgb}{0.2,0.2,0.9}
\definecolor{superdarkblue}{rgb}{0.2,0.2,0.3}
\hypersetup{colorlinks=true, linkcolor=darkblue,citecolor=darkblue,
urlcolor = superdarkblue}
\definecolor{citegreen}{rgb}{0.2,0.2,0.6}

\usepackage{geometry}

\input commands.tex

\makeatother

\begin{document}

\title[The lowest eigenvalue of a soft quantum ring]{
\textsc{Optimization of the lowest eigenvalue of a soft quantum ring}}

\author{Pavel Exner}
\address{Department of Theoretical Physics, Nuclear Physics Institute, Czech Academy of Sciences, 25068 \v Re\v z near Prague, Czechia, and Doppler Institute for Mathematical Physics and Applied Mathematics, Czech Technical University, B\v rehov\'a 7, 11519 Prague, Czechia}
\email{exner@ujf.cas.cz}

\author{Vladimir Lotoreichik}
\address{Department of Theoretical Physics, Nuclear Physics Institute, 	Czech Academy of Sciences, 25068 \v Re\v z, Czechia}
\email{lotoreichik@ujf.cas.cz}

\begin{abstract}
	We consider the self-adjoint two-dimensional Schr\"odinger operator $\Op$ associated with the differential expression $-\Delta -\mu$ describing a particle exposed to an attractive interaction given by a measure $\mu$ supported in a closed curvilinear strip and having fixed transversal one-dimensional profile measure $\mu_\bot$. This operator has nonempty negative discrete spectrum and we obtain two optimization results for its lowest eigenvalue. For the first one, we fix $\mu_\bot$ and maximize the lowest eigenvalue with respect to shape of the curvilinear strip the optimizer in the first problem turns out to be the annulus. We also generalize this result to the situation which involves an additional perturbation of $\Op$ in the form of a positive multiple of the characteristic function of the domain surrounded by the curvilinear strip.  Secondly, we fix the shape of the curvilinear strip and minimize the lowest eigenvalue with respect to variation of $\mu_\bot$, under the constraint that the total profile measure $\aa >0$ is fixed. The optimizer in this problem is $\mu_\bot$ given by the product of $\aa$ and the Dirac $\dl$-function supported at an optimal position.
\end{abstract}

\maketitle

\section{Introduction}
Spectral properties of Schr\"odinger operators describing particles localized to tubular regions attracted a lot of attention. There is more than one reason for that. One motivation comes from physics where such operators serve as models of guided quantum dynamics; the localization at that may be realized in different ways, either by Dirichlet conditions corresponding to hard walls, or by potentials, regular or singular, supported in the vicinity of a curve or another submanifold \cite{EK15}. On the other hand, there are many interesting mathematical problems here revealing intricate relations between spectra and the geometry of the interaction support. An important place among them belongs to optimization problems to which the topic of this paper pertains. In brief, our motivation here is twofold. On the one hand, we are going to show that the ground state maximization in ring-shaped structures known to be valid for Dirichlet strips \cite[Sec.~3.2.3]{EK15}, Robin strips~\cite{EL20}, and `leaky wires' \cite{EHL06} also takes place for `soft' quantum rings reminiscent of the `soft' waveguides considered recently \cite{E20, KKK20}. On the other hand, we intend to demonstrate a more general approach to the problem allowing one to treat such Schr\"odinger operators on the same footing for both regular and singular potentials. This will also make it possible to address the question about the `transverse optimization' in such problems which so far escaped attention.

To be more specific, in the present paper we deal with the spectral optimization for the two-dimen\-sional self-adjoint Schr\"odinger operator $\sfH_\mu$ in the Hilbert space $L^2(\dR^2)$ corresponding to the formal differential expression $-\Delta  - \mu$ with a measure $\mu$ on $\dR^2$, where the minus sign in front of $\mu$ means that the interaction is attractive. Recall that Schr\"odinger operators with interactions given by measures are considered in \eg~\cite{BEKS94, BFT98, KL14, V09}. Working with measure-type potentials serves the goal indicated above, namely to provide a unified description for regular potentials as well as for $\dl$-potentials supported on manifolds of codimension one. A natural way to introduce two-dimensional Schr\"odinger operators with attractive interactions given by a reasonably wide class of measures is to define such an operator by means of a closed, densely defined, symmetric and semi-bounded sesquilinear form in the Hilbert space $L^2(\dR^2)$,
\[
	H^1(\dR^2)\ni u \mapsto \int_{\dR^2}|\nabla u|^2\dd x -\int_{\dR^2}|u|^2\dd\mu.
\]

Our optimization results are formulated for measures supported in a closed generic curvilinear strip, which is defined as the set of all points whose distance from a given $C^2$-smooth closed curve $\Sg\subset\dR^2$ does not exceed a parameter $d_- >0$ for points surrounded by $\Sg$ and a parameter $d_+ > 0$ for points outside $\Sg$. The parameters $d_\pm$ are chosen small enough so that inside the strip the parallel coordinates $(s,t)$ based on the distance from $\Sg$ are globally well defined, where $s$ is the longitudinal variable and $t$ is the transversal variable; see Subsection~\ref{ssec:strips}.

Moreover, we restrict our attention to  measures $\mu$ of a special structure given in parallel coordinates by
\begin{equation}\label{eq:mu_intro}
	\dd\mu = (1+\kp(s)t)\dd s\dd\mu_\bot(t),
\end{equation}
where $\kp$ is the curvature of $\Sg$ and $\mu_\bot$ is a finite measure on the interval $[-d_-,d_+]$; see Subsection~\ref{ssec:measure} for details.
If the measure $\mu_\bot$ is generated by a bounded real-valued potential, $\dd\mu_\bot = w_\bot(t)\,\dd t$, the Hamiltonian $\Op$ models a soft quantum ring in the spirit of \cite{E20, KKK20}. On the other hand, in the case of $\dd\mu_\bot = \aa\dl_0$, where $\aa >0$ is a coupling constant and $\dl_0$ is the Dirac $\dl$-function supported at the origin, the potential-measure $\mu$ reduces to the distributional potential $\aa\dl_\Sg$, where $\dl_\Sg$ is the $\dl$-function supported on $\Sg$. The literature on Schr\"odinger operators with $\dl$-interactions supported on curves is vast; see the review paper~\cite{E08}, the monograph~\cite[Chap.~10]{EK15}, the references therein, and also more recent contributions, \eg~\cite{BLL13, DEKP16, G20, LO16, MP19}.

It can be shown that the essential spectrum of $\Op$ with $\mu$ as in~\eqref{eq:mu_intro} coincides with the positive semi-axis and that there is at least one negative eigenvalue. In the present paper, we establish two spectral shape optimization results for the lowest negative eigenvalue of the operator $\Op$.

The first result, given in Theorem~\ref{thm1}, states that the lowest eigenvalue is maximized by the potential supported in the circular strip (the annulus) provided that we fix the length of the curve on which the strip is constructed and also fix the transversal profile $\mu_\bot$ of the measure $\mu$ in~\eqref{eq:mu_intro}. This is new also in the particular case of regular potentials; for $\dl$-potentials supported on curves it reduces to~\cite[Thm. 4.1]{EHL06} providing an alternative proof of the said result. Our, more general proof relies on the fact that the ground-state for the circular strip is a radial function. As test functions in the min-max principle for $\sfH_\mu$ corresponding to the generic curvilinear strip we take transplantations of smooth, compactly supported, and radially symmetric functions approximating the ground state for the circular strip. The transplantation is performed by means of the parallel coordinates on $\dR^2$, which are associated with the distance function from $\Sg$.
Under this operation the kinetic energy term $\int_{\dR^2}|\nabla u|^2\dd x$ and the $L^2$-norm of the function do not increase while the potential energy term $\int_{\dR^2}|u|^2\dd\mu$ turns out to be preserved. This proof technique is inspired by related optimization for the lowest negative eigenvalue of the Robin Laplacian~\cite{AFK17, BFNT18, FK15, KL18, KL20}.

Moreover, we provide in Theorem~\ref{thm1b} a generalization of this optimization result in which $\Op$ is amended by an additive perturbation by a positive multiple of the characteristic function of the bounded open set surrounded by the curvilinear strip. In this more general setting the lowest negative eigenvalue need not always exist if the added steplike potential is too large, but the maximizer of the lowest spectral point remains to be given by the radially symmetric configuration.

As for the second question, we fix the curvilinear strip underlying the measure $\mu$, fix the parameter $\aa := \mu_\bot([-d_-,d_+])$ and prove in Theorem~\ref{thm2} that the lowest eigenvalue of $\Op$ is minimized by the transversal potential $\mu_\bot = \aa\dl_{t_\star}$ being the $\dl$-function supported at an optimal position $t_\star\in [-d_-,d_+]$. The  optimal position need not coincide with a boundary point of the interval $[-d_-,d_+]$ and in general we can not establish its precise location, apart from the special case of an annulus; \cf Remark~\ref{rem:annulus}. This result is reminiscent of the optimization for the lowest eigenvalue of the one-dimensional Schr\"odinger operator on an interval with fixed $L^1$-norm of an attractive potential~\cite{T84}.

\subsection*{Organization of the paper}
In Section~\ref{sec:prelim} we provide a preliminary material, which is needed to formulate and discuss in Section~\ref{sec:main} the main results of the paper. Section~\ref{sec:annulus} is devoted to the explicit analysis of the interaction given by a radially symmetric measure supported in an annulus. The method of parallel coordinates is outlined in Section~\ref{sec:parallel}. Proofs of the main results are given in Section~\ref{sec:proofs}. The paper is complemented by Appendix~\ref{app:trace} in which we
prove a special version of the trace theorem
and relying on it analyze a family of auxiliary operators with $\dl$-interactions.

\section{Preliminaries}\label{sec:prelim}
The material needed in the following is split here in three subsections. First, in Subsection~\ref{ssec:strips} we specify the geometric setting. Secondly, in Subsection~\ref{ssec:measure}, we define a family of measures. Finally, in Subsection~\ref{ssec:hamiltonian} we introduce a class of Schr\"odinger operators with potentials given by measures and analyze its basic properties.

\subsection{Curvilinear strips}\label{ssec:strips}

Let $\Sg\subset\dR^2$ be a closed $C^2$-smooth curve of length \mbox{$L>0$}. We implicitly assume that $\Sg$ is the boundary of a simply-connected bounded planar domain $\Omg_+\subset\dR^2$. By $\Omg_- := \dR^2\sm\ov{\Omg_+}$ we denote the complement of $\Omg_+$. In the following, $\nu$ stands for the outer unit normal vector to the domain $\Omg_+$. We parametrize the curve $\Sg$ counterclockwise by its arc length, that is, using the map $\s\colon [0,L]\arr\dR^2$ for which the tangential vector $\tau(s) := \dot\s(s)$ is of unit length. With a slight abuse of notation we use the abbreviation $\nu(s) = \nu(\s(s))$. We denote by $\kp\colon[0,L]\arr\dR$ the curvature of $\Sg$, the sign of which is chosen so that $\kp$ is non-negative provided that $\Omg_+$ is convex. The Frenet-Serret formula connects $\tau$, $\nu$, and $\kp$ by
\begin{equation}\label{eq:FS}
	\dot\tau(s) = -\kp(s)\nu(s).
\end{equation}

For $d_\pm \ge 0$ with $d_+ +d_- > 0$, we set $\cI :=[-d_-,d_+]$ and consider the mapping
\begin{equation}\label{eq:mapping}
	[0,L)\times \cI \ni (s,t) \mapsto \s(s) + t\nu(s).
\end{equation}
According to~\cite[Thm. 5.25]{Lee} (see also~\cite[Prop. B.2]{BEHL17}) there exist numbers $D_- = D_-(\Sg) > 0$ and $D_+ = D_+(\Sg) \in (0,\infty]$ such that the mapping~\eqref{eq:mapping} is injective for all $d_+ < D_+$ and $d_- < D_-$. In what follows, we assume that $D_\pm$ are chosen to be the largest possible. We note that~\eqref{eq:FS} and the injectivity of the mapping~\eqref{eq:mapping} imply that its Jacobian $J(s,t)$  satisfies
\begin{equation}\label{eq:Jac}
	J(s,t) = 1+t\kp(s) > 0\qquad\text{for	all}\quad s\in [0,L)\quad \text{and}\quad t \in (-D_-,D_+).
\end{equation}
It is also worth to mention that $D_+(\Sg) = \infty$ provided that $\Omg_+$ is a convex domain.

We define the closed curvilinear strip $\Pi_\cI(\Sg)\subset\dR^2$ with $d_\pm\in(0,D_\pm)$ of the width $d_-+d_+ > 0$ as follows,
\begin{equation}\label{key}
	\Pi_\cI(\Sg) := \big\{\s(s) + t\nu(s)\in\dR^2\colon s\in [0,L), t\in \cI\big\}.
\end{equation}
Furthermore, we specify the domain surrounded by the strip $\Pi_\cI(\Sg)$ as
\begin{equation}\label{eq:Omg}
	\Omg := \Omg_+\sm\Pi_\cI(\Sg)
\end{equation}
and denote by $\chi$ its characteristic function. For a fixed $t\in\cI$ we introduce the curve
\begin{equation}\label{eq:Sgt}
	\Sg_t := \{\s(s)+t\nu(s)\in\dR^2\colon s\in[0,L)\}
\end{equation}
located at the distance $|t|$ from $\Sg$, being inside $\Omg_+$ for $t < 0$ and outside it for $t > 0$. The curve $\Sg_0$ can be identified with $\Sg$.

Let $\cC\subset\dR^2$ be a circle of length $L > 0$. The radius of $\cC$ is denoted by $R$ and the identity $R = \frac{L}{2\pi}$ clearly holds. The unit speed counterclockwise parametrization of $\cC$ is given by the map $\s_\circ\colon [0,L]\arr\dR^2$ and the respective outer unit normal is denoted by $\nu_\circ$. In this particular case, we clearly have $D_-(\cC) = R$, $D_+(\cC) = \infty$ and the circular strip $\Pi_\cI(\cC)$ is just an
annulus. The domain $\Omg_\circ$ surrounded by the annulus $\Pi_\cI(\cC)$ is the disk centred at the origin and having the radius $R - d_-$. We denote its characteristic function by $\chi_\circ$.

\subsection{The class of finite measures}
\label{ssec:measure}

Let the geometric setting be as in Subsection~\ref{ssec:strips}. The perturbations considered in this paper, in general singular, are associated to a class of Radon measures on $\dR^2$ which we introduce here and specify its particular cases. In the following, we denote by $\one_\cU$ the characteristic function of an open set $\cU\subset\dR^2$.

Consider first a nonzero and finite measure $\mu_\bot$ on the interval $\cI =[-d_-,d_+]$ with $d_\pm \in (0,D_\pm)$. Using it we rigorously define the measure $\mu$ in~\eqref{eq:mu_intro} supported on $\Pi_\cI(\Sg)$ as follows,
\begin{equation}\label{eq:potentials}
	\mu(\cU) := \int_0^L\int_{-d_-}^{d_+}
	\one_{\cU}\big(\s(s) +t\nu(s)\big)\,
	(1+\kp(s)t) \dd \mu_\bot(t)\dd s
\end{equation}
for an open set $\cU\subset\dR^2$. Taking the structure of the measure $\mu$ into account, the one-dimensional measure $\mu_\bot$ will be occasionally called transversal.

In the class of measures~\eqref{eq:potentials} we single out several particular cases. First of all, we introduce the measure $\mu_\circ$ which corresponds to the annulus $\Pi_\cI(\cC)\subset\dR^2$
\begin{equation}\label{eq:potentials2}
\mu_\circ(\cU) := \int_0^L\int_{-d_-}^{d_+}
\one_{\cU}\big(\s_\circ(s) +t\nu_\circ(s)\big)\,
\left(1+\frac{t}{R}\right) \dd \mu_\bot(t)\dd s,
\qquad \cU\subset\dR^2.
\end{equation}
Secondly, we consider the case when the transversal measure is given by $\dd\mu_\bot = w_\bot(t)\dd t$ with a real-valued non-negative $w_\bot \in L^\infty(\cI)$. In this case the measure $\mu$ in~\eqref{eq:potentials} can be characterized by its Radon-Nikodym derivative, $\dd\mu = V(x)\dd x$, where
\begin{equation}\label{eq:V}
V(x) = \begin{cases} 0,&\quad \text{if}\;\; x\notin \Pi_\cI(\Sg),\\
w_\bot(t), &\quad \text{if}\;\; x = \s(s) + t\nu(s)\in \Pi_\cI(\Sg).
\end{cases}		
\end{equation}
Finally, we single out the case when the transversal measure is given by $\mu_\bot = \aa\dl_t$, where $\aa$ is a positive real number and $\dl_t$ is the one-dimensional Dirac $\delta$-function supported at the point $t\in\cI$. In this case the measure $\mu$ in~\eqref{eq:potentials} can be identified with
$\aa\dl_{\Sg_t}$, where $\dl_{\Sg_t}$ is the $\dl$-function supported on the closed curve $\Sg_t$ defined in~\eqref{eq:Sgt}.

\subsection{The Hamiltonian}\label{ssec:hamiltonian}

Let the measure $\mu$ be as in Subsection~\ref{ssec:measure}. The operator we are interested in is associated to the formal differential expression $-\Delta - \mu$ on $\dR^2$. In order to introduce it properly, we use the form approach.

\begin{prop}
The quadratic form
\begin{equation}\label{eq:form}
	\frh_{\mu}[u] := \|\nabla u\|^2_{L^2(\dR^2;\dC^2)} -
	\int_{\dR^2}|u|^2\dd \mu,\qquad
	\dom\frh_{\mu} := H^1(\dR^2),
\end{equation}
is closed, densely defined, symmetric, and lower-semibounded in the Hilbert space $L^2(\dR^2)$.
\end{prop}
\begin{proof}
	It is clear that the quadratic form $\frh_\mu$
	is symmetric. Moreover, since the Sobolev space $H^1(\dR^2)$ is dense in $L^2(\dR^2)$, the quadratic form $\frh_\mu$ is also densely defined. It remains to show that the form $\frh_\mu$ is closed and lower-semibounded.
	By Lemma~\ref{lem:trace}, proven in Appendix~\ref{app:trace}, for any $\eps' > 0$ there is a
	constant $C'(\eps') > 0$ such that for any $t\in \cI$ the following inequality
	\[
		\|u|_{\Sg_t}\|^2_{L^2(\Sg_t)} \le
		\eps'\|\nabla u\|^2_{L^2(\dR^2;\dC^2)} + C'(\eps')\|u\|^2_{L^2(\dR^2)}
	\]
	holds for all $u\in H^1(\dR^2)$. It is essential here that $C'$ is independent of $t$.	
	Hence, we get the following estimate
	\[
	\begin{aligned}
		\int_{\dR^2}|u|^2\dd \mu
		& = \int_{-d_-}^{d_+}\int_0^L |u(\s(s)+t\nu(s))|^2(1+\kp(s)t)\dd s \,\dd\mu_\bot(t) \\
		&=
		\int_{-d_-}^{d_+}\|u|_{\Sg_t}\|^2_{L^2(\Sg_t)} \,\dd\mu_\bot(t)
		\le \eps'\mu_\bot(\cI)\|\nabla u\|^2_{L^2(\dR^2;\dC^2)} + C'(\eps')\mu_\bot(\cI)\|u\|^2_{L^2(\dR^2)}
	\end{aligned}
	\]
	for all $u\in H^1(\dR^2)$.
	Setting $\eps' =\frac{\eps}{\mu_\bot(\cI)}$ in the above inequality we infer that for all $\eps > 0$ there exists a constant, explicitly given by $C(\eps) := C'(\frac{\eps}{\mu_\bot(\cI)})\mu_\bot(\cI)$,
	such that
	\begin{equation}\label{eq:muineq}
		\int_{\dR^2}|u|^2\dd\mu \le \eps\|\nabla u\|^2_{L^2(\dR^2;\dC^2)} + C(\eps)\|u\|^2_{L^2(\dR^2)}
	\end{equation}
	for all $u\in H^1(\dR^2)$. This means that the densely defined symmetric form $\frh_\mu$ is form bounded with respect to the closed, densely defined, symmetric, and lower semi-bounded form $H^1(\dR^2)\ni u\mapsto \|\nabla u\|^2_{L^2(\dR^2;\dC^2)}$ with the form bound $< 1$,  and by~\cite[Thm. VI.1.33]{Kato} it follows that the form $\frh_\mu$ is closed and semi-bounded as well.
\end{proof}
\begin{dfn}\label{def:Op}
	The self-adjoint Schr\"odinger operator $\Op$ in the Hilbert space $L^2(\dR^2)$ is associated to the form $\frh_{\mu}$ in~\eqref{eq:form} via the first representation theorem~\cite[Thm. VI.2.1]{Kato}.
\end{dfn}
The essential spectrum of $\Op$ can be characterised explicitly and does not depend on the measure $\mu$.
\begin{prop}\label{prop:ess}
	The essential spectrum of $\Op$ coincides with $[0,\infty)$.
\end{prop}
\begin{proof}
	In view of inequality~\eqref{eq:muineq} the measure $\mu$ belongs to the class considered in~\cite{BEKS94}. Since the measure $\mu$ is finite, the characterisation of the essential spectrum follows directly from~\cite[Thm. 3.1]{BEKS94}.
\end{proof}
Furthermore, the criticality of the Laplacian in two dimensions yields the following property of the discrete spectrum of $\Op$.
\begin{prop}
	The discrete spectrum of $\Op$ in $(-\infty,0)$ is non-empty.
\end{prop}
\begin{proof}
	The argument relies on the construction of an appropriate sequence of test functions for the min-max principle. This sequence $\{\varphi_n\}_{n\in\dN}$ is explicitly given by
	\[
		\varphi_n(x) = \begin{cases}
		1 &\quad \text{if}\;\; |x| < n,\\
		\frac{\log n^2 - \log|x|}{\log n^2-\log n}&\quad \text{if}\;\; n \le |x|< n^2,\\
		0 &\quad \text{if}\;\; |x| \ge n^2.
		\end{cases}
	\]
	Clearly, we have $\varphi_n\in H^1(\dR^2)$ for all $n\in\dN$.
	Next, we observe that
	\begin{equation}\label{eq:nablaphin}
		\|\nabla \varphi_n\|^2_{L^2(\dR^2;\dC^2)} =
		\frac{2\pi}{(\log n)^2}
		\int_{n}^{n^2}\frac{1}{r}\,\dd r = \frac{2\pi}{\log n}\arr 0 \quad\text{holds as}\; n\arr\infty,
	\end{equation}
	and moreover, since $\mu$ is supposed to be nonzero, we have for all $n$ large enough
	\begin{equation}\label{eq:muphin}
		\int_{\dR^2}|\varphi_n|^2\,\dd \mu =
		\int_{\dR^2}\dd \mu = \mu(\dR^2) > 0.
	\end{equation}
	Combining~\eqref{eq:nablaphin} and~\eqref{eq:muphin}, we conclude that
	$\frh_\mu[\varphi_n] < 0$ holds for all $n\in\dN$ large enough. In view of Proposition~\ref{prop:ess}, it follows then from the min-max principle that the negative discrete spectrum of $\Op$ is non-empty.
\end{proof}
We denote by $\lm_1(\mu) < 0$ the lowest eigenvalue of $\Op$. Using again the min-max principle, we can characterize this eigenvalue as
\begin{equation}\label{eq:lm1}
	\lm_1(\mu) =
	\inf_{u\in H^1(\dR^2)\sm\{0\}}
	\frac{\frh_\mu[u]}{\|u\|^2_{L^2(\dR^2)}}.
\end{equation}

We remark that for $\dd\mu(x) = V(x)\dd x$ with $V$ given in~\eqref{eq:V} the Hamiltonian $\Op$ can be alternatively characterized via its action and the operator domain,
\[
	H^2(\dR^2)\ni u\mapsto -\Delta u - Vu.
\]
Similar Schr\"odinger operators have been considered recently~\cite{E20, KKK20} as a tool to treat soft quantum waveguides.

Concerning the other particular case, $\mu = \aa\delta_\Sg$ with $\aa > 0$, according to~\cite[Def. 3.4, Thm. 3.6]{BLL13}  the Hamiltonian $\sfH_{\aa\dl_\Sg}$ can be also written as a Schr\"odinger operator, this time with $\dl$-interaction of strength $\aa > 0$ supported by the curve $\Sg$. Its action and operator domain are
\[
\{u\in (H^2(\Omg_+)\oplus H^2(\Omg_-))\cap H^1(\dR^2)\colon
[\p_{\nu}u]_{\Sg}
=\aa u|_{\Sg}\}
\ni u\mapsto (-\Delta u_+)\oplus (-\Delta u_-),
\]
where $u_\pm := u|_{\Omg_\pm}$ and $[\p_{\nu}u]_{\Sg} = \nu\cdot(\nabla u_+)|_{\Sg} - \nu\cdot(\nabla u_-)|_\Sg$ stands for the jump of the normal derivative across the interface $\Sg$. This corresponds to the formal differential expression $-\Delta -\aa\delta_\Sg$.

Recall that $\chi$ denotes the characteristic function of the bounded open set $\Omg$ in~\eqref{eq:Omg} surrounded by the curvilinear strip $\Pi_\cI(\Sg)$. Let the coupling constant $\beta > 0$ be arbitrary. Along with the Hamiltonian $\Op$ we consider its bounded additive perturbation $\Op + \beta\chi$. 
According to~\cite[Thm. 3.1]{BEKS94} the essential spectrum of $\Op+\beta\chi$ is the same as of $\Op$ and coincides with $[0,\infty)$. 
We denote by $\lm_1^\beta(\mu)\le0$ the lowest spectral point of $\Op+\beta\chi$ which, in dependence on the value of $\beta$ and the measure $\mu$ is either a negative eigenvalue or zero representing the bottom of the essential spectrum. 
The point $\lm_1^\beta(\mu)$ admits the following variational characterization,
\begin{equation}\label{eq:lm1beta}
	\lm_1^\beta(\mu) =
	\inf_{u\in H^1(\dR^2)\sm\{0\}}
	\frac{\frh_\mu[u] +\beta\int_\Omg|u|^2\dd x}{\|u\|^2_{L^2(\dR^2)}}.
\end{equation}
By~\eqref{eq:lm1beta}, the function $\beta\mapsto \lm_1^\beta(\mu)$ is non-decreasing and it may happen that $\lambda^\beta_1(\mu)=0$; we know, for instance, that for a sign indefinite potential the discrete spectrum is empty in the weak coupling regime provided the integral of the potential is positive~\cite{S76}.
\section{Main results}\label{sec:main}

Before stating the results we have to spell out the assumptions.
\begin{hyp}\label{hyp:main}
	Let a $C^2$-smooth curve $\Sg$ and the circle $\cC$ be such that $|\Sg| = |\cC| = L >0$. Let $d_\pm \in (0,D_\pm(\Sg))$ and the transversal Radon measure $\mu_\bot$
	on $\cI = [-d_-,d_+]$, nonzero and finite, be fixed.
	Let further the measure $\mu$ be associated with $\Sg$ and $\mu_\bot$ as in \eqref{eq:potentials}
	and let the measure $\mu_\circ$ be associated with $\cC$ and $\mu_\bot$ as in~\eqref{eq:potentials2}.
	Finally, let the Schr\"odinger operators $\Op$ and $\sfH_{\mu_\circ}$ be as in Definition~\ref{def:Op}.
\end{hyp}

In the first main result we maximize the lowest eigenvalue of $\Op$ with respect to variation of the shape of the curvilinear strip supporting the measure $\mu$, while the transversal measure $\mu_\bot$ remains fixed.

\begin{thm}\label{thm1}
	Assume that Hypothesis~\ref{hyp:main} holds. Then the lowest negative eigenvalues $\lm_1(\mu)$ and $\lm_1(\mu_\circ)$, respectively, of $\Op$ and of $\sfH_{\mu_\circ}$ satisfy the inequality
	\[
		\lm_1(\mu) \le \lm_1(\mu_\circ).
	\]
\end{thm}
The proof of Theorem~\ref{thm1} relies on the min-max principle applied on the level of quadratic forms. Appropriate test functions for the variational characterization of $\lm_1(\mu)$ in~\eqref{eq:lm1} are constructed by means of transplantations of smooth compactly supported approximations of the radial ground state of $\sfH_{\mu_\circ}$.
This transplantation is performed using the method of parallel coordinates and the transplanted functions depend essentially on the distance to $\Sg$. In the course of the analysis we employ the co-area formula to prove that under transplantation the kinetic energy $\int_{\dR^2}|\nabla u|^2\dd x$ and the $L^2$-norm of the function do not increase, and we use the total curvature identity to show that the potential energy  $\int_{\dR^2}|u|^2\dd\mu$ is preserved. This strategy of the proof is simple but powerful; it was recently successfully applied to the Robin Laplacian on bounded domains~\cite{AFK17, BFNT18, FK15}, on 2-manifolds~\cite{KL19}, and on exterior domains~\cite{KL18, KL20}, and also to the two-dimensional Schr\"odinger operator with $\dl'$-interaction supported on a closed curve~\cite{L18}.
\begin{remark}
	In the particular case $\mu_\bot = \aa\delta_0$ the claim of Theorem~\ref{thm1} implies the spectral isoperimetric inequality for Schr\"odinger operators with $\dl$-interactions supported on closed curves of fixed length~\cite[Thm. 4.1]{EHL06}, proven there by a different method, via the Birman-Schwinger principle.
\end{remark}
\begin{remark}
	We conjecture that the strict inequality, $\lm_1(\mu) < \lm_1(\mu_\circ)$, holds provided that $\Sg$ is not congruent with $\cC$.
\end{remark}

As a variation on the first main result, we slightly generalize Theorem~\ref{thm1} as follows:
\begin{thm}\label{thm1b}
	Adopt again Hypothesis~\ref{hyp:main}. Let the coupling constant $\beta > 0$ be fixed and let the characteristic functions $\chi$ and $\chi_\circ$ be as in Subsection~\ref{ssec:strips}. Then the lowest spectral points $\lm_1^\beta(\mu)$ and $\lm_1^\beta(\mu_\circ)$, respectively, of $\Op +\beta\chi$ and of $\sfH_{\mu_\circ}+\beta\chi_\circ$ satisfy the inequality
	\[
		\lm_1^\beta(\mu) \le \lm_1^\beta(\mu_\circ).
	\]
\end{thm}
The proof of Theorem~\ref{thm1b} relies on the same argument as in the proof of Theorem~\ref{thm1}. In addition, we make use of the fact that the potential energy term in the variational characterization of the lowest spectral points corresponding to the perturbation is not made larger by
the transplantation. This is no longer true for $\beta < 0$, of course, thus the positivity of $\beta$ is essential. Theorem~\ref{thm1} can be viewed as a particular case of Theorem~\ref{thm1b} upon including the case $\beta = 0$. For the reader's convenience we decided to state these results separately.
\begin{remark}
	As a consequence of Theorem~\ref{thm1b} we conclude that if the (negative) discrete spectrum of $\sfH_{\mu_\circ}+\beta\chi_\circ$ is nonempty, the same is true for $\Op+\beta\chi$.
\end{remark}

The second main result concerns minimization of the lowest eigenvalue of $\Op$ with respect to variation of
the transversal measure $\mu_\bot$, while preserving the shape of the curvilinear strip supporting the measure $\mu$.
\begin{thm}\label{thm2}
	Assume that Hypothesis~\ref{hyp:main} holds. Let the closed curves $\Sg_t$ be defined as in~\eqref{eq:Sgt}.
	Put $\aa = \mu_{\bot}(\cI)$ and suppose that the Schr\"odinger operators and $\sfH_{\aa\dl_{\Sg_t}}$, $t\in\cI$, are as in Definition~\ref{def:Op} with $\mu = \aa\dl_{\Sg_t}$. Then the lowest eigenvalues $\lm_1(\mu)$ and $\lm_1(\aa\dl_{\Sg_t})$ of $\Op$ and of $\sfH_{\aa\dl_{\Sg_t}}$, respectively, satisfy the inequality
	\begin{equation}\label{eq:ineq2}
		\lm_1(\mu) \ge \min_{t\in\cI}\lm_1(\aa\dl_{\Sg_t}).
	\end{equation}
\end{thm}

\smallskip

\begin{remark}
	In view of Lemma~\ref{lem:Op_cont} the function
	$\cI\mapsto \lm_1(\aa\dl_{\Sg_t})$
	is continuous on the closed interval $\cI$, hence the minimum in~\eqref{eq:ineq2} exists being attained at some point $t_\star \in\cI$.
\end{remark}
The proof of Theorem~\ref{thm2} relies on the min-max principle again. As the test function for the operator $\sfH_{\aa\dl_{\Sg_t}}$ we now take the ground state of $\sfH_\mu$. The construction is inspired by the note~\cite{T84} on the optimization of the eigenvalue of the one-dimensional Schr\"odinger operator on an interval with respect to attractive potential with a prescribed $L^1$-norm.

\begin{remark}\label{rem:annulus}
	According to the numerical computation in~\cite{ET04} and the rigorous analysis in~\cite[Prop. 5]{OPP18},
	there is $R_\star = R_\star(\aa) > 0$ such that the function
	$f_\aa(R) := \lm_1(\aa\dl_\cC)$
	of the radius $R$ of the circle $\cC$ is decreasing on $[0,R_\star]$ and increasing on $[R_\star,\infty)$ with $\lim_{R\to\infty}\lm_1(\aa\dl_\cC) = -\frac14\aa^2$. Moreover, let $t_\star \in \cI$ be a value of $t$ at which the minimum on the right-hand side of~\eqref{eq:ineq2} is attained.
	In the circular case it is easy to see that
	\[
		t_\star =
		\begin{cases}
		-d_-,
			  &\quad \text{if}\;\; R_\star < R - d_-,\\
		R_\star - R,
			  &\quad \text{if}\;\; R_\star \in (R-d_-,R+d_+),\\
		d_+, &\quad \text{if}\;\; R_\star > R+d_+.
		\end{cases}
	\]
	This shows that the optimal transverse measure need not correspond to an endpoint of the interval $\cI$.	
\end{remark}

\section{Soft quantum annulus}\label{sec:annulus}

As the first ingredient of the proof we analyse the ground-state corresponding to the Schr\"odinger operator $\sfH_{\mu_\circ}$ as in Definition~\ref{def:Op} with the potential given by~\eqref{eq:potentials2}. Let $(r,\tt)$ be the polar coordinates with the pole at the center of the circle $\cC$.
\begin{prop}\label{prop:annulus}
	Assume that Hypothesis~\ref{hyp:main} holds.
	A normalized eigenfunction $u_1^\circ(r,\tt) = \psi_\circ(r)$ corresponding to the lowest eigenvalue $\lm_1(\mu_\circ)$ of $\sfH_{\mu_\circ}$ is independent of the angular variable and $\lm_1(\mu_\circ)$ admits the following variational characterization
	\begin{equation}\label{eq:lm1var}
	\lm_1(\mu_\circ) =
		\inf_{
		\stackrel[\psi\not=0]{}{\psi \in C^\infty_0([0,\infty))}}
	\frac{\displaystyle\int_0^\infty
		|\psi'(r)|^2 \,r\,\dd r -
		\int_{-d_-}^{d_+} |\psi(R+t)|^2\left(R+t\right)\,\dd\mu_\bot(t)}{\displaystyle
		\int_0^\infty |\psi(r)|^2\,r\,\dd r}.
	\end{equation}
\end{prop}
\begin{proof}
	The proof relies on the separation of variables in polar coordinates. The space $L^2(\dR^2)$ is isomorphic to	
	\[
		L^2_{\rm pol}(\dR^2) := L^2(\dR_+\times [0,2\pi);r\,\dd r\dd \tt)
	\]
	with the associated first-order Sobolev space
	\[
		H^1_{\rm pol}(\dR^2) := \{u\in L^2_{\rm pol}(\dR^2)\colon|\nabla_{\rm pol} u|\in L^2_{\rm pol}(\dR^2)\},
	\]
	where $\nabla_{\rm pol}$ stands for the gradient in polar coordinates. By change of variables one can rewrite the quadratic form $\frh_{\mu_\circ}$ as
	\[
		\frh_{\mu_\circ}^{\rm pol}[u] = \int_0^{2\pi}\int_0^\infty\left[|\p_r u|^2
		+ \frac{|\p_\tt u|^2}{r^2}\right]r\,\dd r\dd \tt - \int_0^{2\pi}\int_{-d_-}^{d_+}|u(R+t,\tt)|^2(R+t)\,\dd\mu_{\bot}(t)\dd\tt.
	\]
	with $\dom\frh^{\rm pol}_{\mu_\circ} = H^1_{\rm pol}(\dR^2)$.
	We employ the complete family of mutually orthogonal `partial wave' projections on $L^2(\dR^2)$,
	\begin{equation}\label{key}
		(\Pi_nu)(r,\tt) = \frac{e^{\ii n\tt}}{2\pi}\int_0^{2\pi}u(r,\tt') \,\mathrm{e}^{-\ii n\tt'}\dd \tt',\quad\; n\in\dZ.
	\end{equation}
	Identifying $\ran\Pi_n$ and $L^2(\dR_+;r\dd r)$, we arrive at the orthogonal decomposition of the Hilbert space $L^2_{\rm pol}(\dR^2)$,
	\[
		L^2_{\rm pol}(\dR^2) \simeq \bigoplus_{n\in\dZ}L^2(\dR_+;r\dd r)
	\]
	and the respective decomposition of the operator $\sfH_{\mu_\circ}$ as the orthogonal sum
	\[
		\sfH_{\mu_\circ} \simeq \bigoplus_{n\in\dZ} \sfH^{[n]}_{\mu_\circ},
	\]
	where the self-adjoint fiber operators $\sfH^{[n]}_{\mu_\circ}$, $n\in\dZ$, in the Hilbert space $L^2(\dR_+;r\dd r)$ are associated with the quadratic forms
	\[
	\begin{aligned}
	\frh^{[n]}_{\mu_\circ}[\psi]&\! :=\!
	\frh_{\mu_\circ}^{\rm pol}\left[
	\frac{\mathrm{e}^{\ii n\tt}\psi(r)}{\sqrt{2\pi}}\right]
	\!=\! \int_0^\infty\left[|\psi'|^2 \!+\!\frac{n^2|\psi|^2}{r^2}\right]r\,\dd r
	\!-\!\int_{-d_-}^{d_+} |\psi(R+t)|^2(R+t)\,\dd\mu_\bot(t),\\
	\dom\frh^{[n]}_{\mu_\circ} &\!:=\!
	\{\psi \in L^2(\dR_+;r\dd r)\colon 	
	\mathrm{e}^{\ii n\tt}\psi(r)\in H^1_{\rm pol}(\dR^2)\} \!=\!
	\Big\{\psi\colon \psi,\psi', \frac{n\psi}{r} \in L^2(\dR_+;r\,\dd r)\Big\}.
	\end{aligned}
	\]
	It is straightforward to see that the form $\frh^{[0]}_{\mu_\circ}$ is smaller in the sense of ordering of the forms than $\frh^{[n]}_{\mu_\circ}$ for any $n\ne 0$, because $\dom\frh^{[0]}_{\mu_\circ}\supset\dom\frh^{[n]}_{\mu_\circ}$ and for any non-trivial $\psi\in\dom\frh^{[n]}_{\mu_\circ}$ the inequality $\frh^{[0]}_{\mu_\circ}[\psi] < \frh^{[n]}_{\mu_\circ}[\psi]$ holds. Hence, an eigenfunction associated with the lowest eigenvalue of $\sfH_{\mu_\circ}$ corresponds to the fiber operator $\sfH_{\mu_\circ}^{[0]}$ and is thus a radial function.
	The variational characterization of $\lm_1(\mu_\circ)$ given in~\eqref{eq:lm1var}
	follows from the min-max principle applied to the quadratic form $\frh_{\mu_\circ}^{[0]}$ on its core $C^\infty_0([0,\infty))$.
\end{proof}
By making minor adjustments in the proof of Proposition~\ref{prop:annulus} we get the following generalization.

\begin{prop}\label{prop:annulus2}
	Adopt again Hypothesis~\ref{hyp:main}. Let the characteristic function $\chi_\circ$ be as in Subsection~\ref{ssec:strips}.
	Assume that $\beta > 0$ is such that
	the lowest spectral point $\lm_1^\beta(\mu_\circ) $
	of $\sfH_{\mu_\circ} +\beta\chi_\circ$ is a negative eigenvalue.
	A normalized eigenfunction $u_{1,
	\beta}^\circ(r,\tt) = \psi^{\beta}_\circ(r)$ corresponding to the lowest eigenvalue $\lm_1^\beta(\mu_\circ)$ of $\sfH_{\mu_\circ}+\beta\chi_\circ$ is independent of the angular variable and $\lm_1^\beta(\mu_\circ)$ admits the following variational characterization
	{\small
	\begin{equation*}
	\lm_1^\beta(\mu_\circ) \!=\!
	\inf_{
	\stackrel[\psi\not=0]{}{\psi \in C^\infty_0([0,\infty))}}
	\frac{\displaystyle\int_0^\infty
		|\psi'(r)|^2 \,r\,\dd r \!-\!
		\int_{-d_-}^{d_+} |\psi(R+t)|^2\left(R+t\right)\,\dd\mu_\bot(t)\!+\!
		\beta\int_0^{R-d_-}|\psi(r)|^2\,r\,\dd r}{\displaystyle
		\int_0^\infty |\psi(r)|^2\,r\,\dd r}.
	\end{equation*}}
\end{prop}

\section{The method of parallel coordinates}
\label{sec:parallel}

The second ingredient of the proof is the method of parallel coordinates. We outline it here referring to the papers~\cite{F41, H64, S01}, as well as to the monograph~\cite{Ba80} and references therein, for further details and proofs.

To begin with, we introduce the distance-functions on the domains $\Omg_\pm$ as
\begin{equation*}\label{key}
\rho_\pm\colon\Omg_\pm\arr\dR_+,\qquad
\rho_\pm(x) := {\rm dist}(x, \Sg).
\end{equation*}
According to, \eg, \cite[Sec. 3]{DZ94} the distance-functions $\rho_\pm$ are Lipschitz continuous with the Lipschitz constant $=1$, differentiable almost everywhere and
\begin{equation}\label{eq:nablarho}
	|\nb\rho_\pm(x)| = 1\qquad \text{for almost all}\,\, x\in\Omg_\pm.
\end{equation}
The set ${\rm Cut}\,(\Omg_\pm)\subset\Omg_\pm$ of zero Lebesgue measure, where the function $\rho_\pm$ is not differentiable, is called the \emph{cut-locus}.

Furthermore, we define the in-radii of $\Omg_\pm$ by
\begin{equation*}\label{eq:in_radius}
R_\pm := \sup_{x\in \Omg_\pm} \rho_\pm(x).
\end{equation*}
The in-radius of $\Omg_+$ is thus the radius of the largest disk in $\dR^2$ that can be inscribed into $\Omg_+$, and due to the well-known isoperimetric inequality, $|\Sg|^2 \ge 4\pi|\Omg_+|$, we get
\begin{align}\label{isop2}
R_+ \le R := \frac{L}{2\pi}.
\end{align}
On the other hand, we obviously  have $R_- = \infty$ in the complement to $\Omg_+$.

Finally, we introduce the following auxiliary functions,
\begin{equation}\label{eq:LmAm}
L_\pm \colon [0, R_\pm] \to \dR_+,
\qquad
L_\pm(t):=
\big|	\{x\in\ov{\Omg_\pm}\colon \rho_\pm(x)= t\}\big|.
\end{equation}
Clearly, $L_\pm(0)= L$ and $L_\pm(t)$ is the length of the corresponding level set of the function $\rho_\pm$.

Let us recall estimates of the functions in~\eqref{eq:LmAm} which will be useful in the sequel, cf. \cite[Prop. A.1]{S01},\cite[Sec. 4]{FK15}, \cite[Sec. 3.1]{KL20}:
\begin{prop}\label{prop_savo}	The functions $L_\pm$ defined by~\eqref{eq:LmAm} satisfy
\[
	L_+(t) \leq L - 2\pi t \and L_-(t) \leq L + 2\pi t.
\]	
\end{prop}
Furthermore, given a real-valued $\psi\in C^\infty_0([0,\infty))$, we introduce the associated functions
$\psi_+ \in C^\infty(([0,R])$ and $\psi_-\in C^\infty_0([0,\infty))$ by
\begin{equation}\label{eq:psipm}
\begin{aligned}
	\psi_+(t) &:= \psi(R-t)\quad&\text{if}\;\;t& \in (0,R),\\ \psi_-(t)&  :=\psi(R+t)\quad&\text{if}\;\;t&\in (0,\infty).
\end{aligned}
\end{equation}
In the proof that will follow we employ test functions of the form
\begin{equation}\label{eq:upsi}
	u_\psi := (\psi_+\circ\rho_+)\oplus (\psi_-\circ\rho_-).
\end{equation}
The properties of $\psi_\pm$ defined in~\eqref{eq:psipm} and the Lipschitz continuity of $\rho_\pm$ imply that
$u_\psi$ is a Lipschitz continuous compactly supported function so that, in particular,  $u_\psi \in H^1(\dR^2\sm\Sg)$. Using $u_{\psi,\pm} := u_\psi|_{\Omg_\pm}$, we note that $u_{\psi,+}|_\Sg = u_{\psi,-}|_\Sg$, and consequently, $u_\psi\in H^1(\dR^2)$.

The co-area formula applied in two dimensions, see~\cite[Thm. 4.20]{B19} and~\cite{MSZ02}, to an open set $\cA\subset\dR^2$, a Lipschitz continuous function $f\colon\cA\arr\dR$, and an integrable function $g\colon\cA\arr\dR$ gives
\begin{equation}\label{eq:coarea}
	\int_\cA g(x)|\nb f(x)|\,\dd x =
		\int_\dR\int_{f^{-1}(t)} g(x)\,\dd \cH^1(x)\,\dd t,
\end{equation}
where $\cH^1$ in the inner integral on the right-hand side is the one-dimensional Hausdorff  measure on the level curve $\{x\in\cA\colon f(x) = t\}$.

In view of~\eqref{eq:nablarho}, we conclude that $|\nb u_{\psi,\pm}| = |\psi'_\pm\circ\rho_\pm|$ almost everywhere in $\Omg_\pm$. Hence, applying the formula~\eqref{eq:coarea} to $g = |\nabla u_{\psi,\pm}|^2$, $f = \rho_\pm$, and $\cA = \Omg_\pm$, using again~\eqref{eq:nablarho} and taking that $R \ge R_+$ into account, we get
\begin{equation}\label{parallel2}
\begin{aligned}
\|\nb u_\psi \|^2_{L^2(\dR^2;\dC^2)}
&\!=\!
\|\nb u_{\psi,+} \|^2_{L^2(\Omg_+;\dC^2)}
+
\|\nb u_{\psi,-} \|^2_{L^2(\Omg_-;\dC^2)}\\
& =
\int_{\Omg_+} |\nb u_{\psi,+}(x)|^2|\nb\rho_+(x)|\,\dd x + \int_{\Omg_-} |\nb u_{\psi,-}(x)|^2|\nb\rho_-(x)|\,\dd x\\
&\!=\!
\int_0^{R_+}|\psi'_+(t)|^2\int_{\rho_+^{-1}(t)} \dd \cH^1(x)  \,\dd t
\!+\!
\int_0^\infty|\psi_-'(t)|^2\int_{\rho_-^{-1}(t)}\dd \cH^1(x)  \,\dd t
\\
&\! =\!
\int_0^{R_+} | \psi_+'(t)|^2
L_+(t) \,\dd t
+
\int_0^\infty | \psi_-'(t)|^2
L_-(t) \,\dd t\\
& \!\le\!
\int_0^{R} | \psi_+'(t)|^2
(L\!-\!2\pi t) \,\dd t
\!+\!
\int_0^\infty | \psi_-'(t)|^2
(L\!+\!2\pi t) \,\dd t\\
&\!=\!
2\pi\int_0^\infty | \psi'(r)|^2
r \dd r,
\end{aligned}
\end{equation}
where Proposition~\ref{prop_savo} was used in the penultimate step. Following the same steps we also get
\begin{equation}\label{parallel3}
\begin{aligned}
\| u_\psi\|^2_{L^2(\dR^2)}
& =
\int_0^{R_+} | \psi_+(t)|^2
L_+(t) \,\dd t
+
\int_0^\infty  |\psi_-(t)|^2
L_-(t) \,\dd t
\\
& \le
\int_0^{R} | \psi_+(t)|^2
(L-2\pi t)\,\dd t
+
\int_0^\infty | \psi_-(t)|^2
(L + 2\pi t) \,\dd t\\
& =2\pi\int_0^\infty|\psi(r)|^2\,r\,\dd r,
\end{aligned}
\end{equation}
and for $\Omg\subset\dR^2$ defined in~\eqref{eq:Omg} we obtain
\begin{equation}\label{parallel3b}
\begin{aligned}
	\int_\Omg |u_\psi|^2\,\dd x
	& =
	\int_{d_-}^{R_+} | \psi_+(t)|^2
	L_+(t) \,\dd t\\
	&\le
	\int_{d_-}^{R} | \psi_+(t)|^2
	(L-2\pi t)\,\dd t =2\pi\int_0^{R-d_-}|\psi(r)|^2\,r\,\dd r.
\end{aligned}	
\end{equation}
Using finally the representation of the measure $\mu$ from~\eqref{eq:potentials} we get
\begin{equation}\label{parallel4}
\begin{aligned}
	\int_{\dR^2}|u_\psi|^2\,\dd\mu&  = \int_0^L\int_{-d_-}^{d_+}
	|u_\psi(\s(t) + t\nu(s))|^2(1+t\kp(s))\,\dd \mu_\bot(t) \dd s \\
	& =
	2\pi
	\int_{-d_-}^{d_+}
	|\psi(R+t)|^2(R+t)\,\dd \mu_\bot(t),
\end{aligned}
\end{equation}
where the total curvature identity $\int_0^{2\pi}\kp(s)\,\dd s = 2\pi$ was used.

\section{Proofs of the main results}\label{sec:proofs}
\subsection{Proofs of Theorems~\ref{thm1} and~\ref{thm1b}}

With all the preparations made above, the arguments are rather short. Recall that we associate with $\psi \in C^\infty_0([0,\infty))$ the function $u_\psi \in H^1(\dR^2)$ as in~\eqref{eq:upsi}. Applying the min-max characterization of $\lm_1(\mu)$, using~\eqref{parallel2},~\eqref{parallel3},~\eqref{parallel4}, and taking into account that $\lm_1(\mu_\circ) < 0$, we find
\[
\begin{aligned}
	\lm_1(\mu) &=
	\inf_{u\in H^1(\dR^2)\sm\{0\}}
	\frac{\frh_\mu[u]}{\|u\|^2_{L^2(\dR^2)}}	\le \inf_{\psi\in C^\infty_0([0,\infty))\sm\{0\}}
	\frac{\frh_\mu[u_\psi]}{\|u_\psi\|^2_{L^2(\dR^2)}}\\
	& \le
	\inf_{\psi\in C^\infty_0([0,\infty))\sm\{0\}}
	\frac{\displaystyle\int_0^\infty|\psi'(r)|^2\,r\,\dd r -
		\int_{-d_-}^{d_+} |\psi(R+t)|^2(R+t)\,\dd \mu_\bot(t)}{\displaystyle\int_0^\infty |\psi(r)|^2\,r\,\dd r} = \lm_1(\mu_\circ),
\end{aligned}	
\]
where the variational characterisation of the eigenvalue $\lm_1(\mu_\circ)$ from Proposition~\ref{prop:annulus} was used in the last step. In this way, Theorem~\ref{thm1} is proved.

Next we pass to Theorem~\ref{thm1b}, the proof of which relies on the same idea. The claim is trivial if $\lm_1^\beta(\mu_\circ) = 0$.
Without loss of generality we may thus assume that the constant $\beta > 0$ and the measure $\mu_\circ$ are such that $\lm_1^\beta(\mu_\circ) < 0$ is a negative eigenvalue of $\sfH_{\mu_\circ} +\beta\chi_{\circ}$. Applying the min-max characterization of $\lm_1^\beta(\mu)$, using~\eqref{parallel2}--\eqref{parallel4}, and taking into account that $\lm_1^\beta(\mu_\circ) < 0$, we find
{\small	
\[
\begin{aligned}
\lm_1^\beta(\mu)
	&=
	\inf_{u\in H^1(\dR^2)\sm\{0\}}
	\frac{\frh_\mu[u]+\beta{\displaystyle\int_\Omg}|u|^2\,\dd x}{\|u\|^2_{L^2(\dR^2)}}	
	\le
	\inf_{\psi\in C^\infty_0([0,\infty))\sm\{0\}}
	\frac{\frh_\mu[u_\psi]+\beta{\displaystyle\int_\Omg}|u_\psi|^2\,\dd x}{\|u_\psi\|^2_{L^2(\dR^2)}}\\
	& \le
	\inf_{
			\stackrel[\psi\not=0]{}{\psi \in C^\infty_0([0,\infty))}}
	\frac{\displaystyle\int_0^\infty|\psi'(r)|^2\,r\,\dd r -
	\int_{-d_-}^{d_+} |\psi(R+t)|^2(R+t)\,\dd \mu_\bot(t) +\beta\int_0^{R-d_-}|\psi(r)|^2r\dd r}{\displaystyle\int_0^\infty |\psi(r)|^2\,r\,\dd r}\\
	& = \lm_1^\beta(\mu_\circ),
\end{aligned}	
\]}
where now the variational characterization of the eigenvalue $\lm_1^\beta(\mu_\circ)$ from Proposition~\ref{prop:annulus2} was used in the last step. This yields the sought claim. \qed

\subsection{Proof of Theorem~\ref{thm2}}

Let $u \in H^1(\dR^2)$ be the normalized ground-state eigenfunction of $\sfH_\mu$. By Lemma~\ref{lem:trace_cont} below the function
$\cI\ni t\mapsto \|u|_{\Sg_t}\|^2_{L^2(\Sg_t)}$ is continuous, and therefore it attains its maximum value at some $t_\star = t_\star(\mu) \in \cI$. In this way, we get
\[
	\int_{\dR^2}|u|^2\dd\mu=
	\int_{-d_-}^{d_+} \|u|_{\Sg_t}\|^2_{L^2(\Sg_t)}\,\dd\mu_\bot(t)\\
	\le
	\|u|_{\Sg_{t_\star}}\|^2_{L^2(\Sg_{t_\star})}\int_{-d_-}^{d_+} \,\dd\mu_\bot(t)
	=
	\aa\|u|_{\Sg_{t_\star}}\|^2_{L^2(\Sg_{t_\star})}.
\]
This allows us to conclude that
\[
\begin{aligned}
	\lm_1(\mu) & = \frh_\mu[u] = \|\nabla u\|^2_{L^2(\dR^2;\dC^2)} - \int_{\dR^2}|u|^2\,\dd\mu\\
	& \ge
	\|\nabla u\|^2_{L^2(\dR^2;\dC^2)}
	-\aa\|u|_{\Sg_{t_\star}}\|^2_{L^2(\Sg_{t_\star})} = \frh_{\aa\dl_{\Sg_{t_\star}}}[u] 	
	\ge \lm_1(\aa\dl_{\Sg_{t_\star}}) \ge \min_{t\in\cI}\lm_1(\aa\dl_{\Sg_t}),
\end{aligned}
\]
where in the last step we used the continuity of $\cI\ni t\mapsto \lm_1(\aa\dl_{\Sg_t})$ proven in Lemma~\ref{lem:Op_cont} below. \qed

\begin{appendix}

\section{Traces on $\Sg_t$ and the lowest eigenvalue of $\sfH_{\aa\dl_{\Sg_t}}$}\label{app:trace}

The first auxiliary result concerns a $t$-independent upper bound on $\|u|_{\Sg_t}\|^2_{L^2(\Sg_t)}$ for an $H^1$-function $u$ in terms of the $L^2$-norms of $u$ itself and its gradient.
\begin{lem}\label{lem:trace}
	Let the curve $\Sg_t$ be as in~\eqref{eq:Sgt}.
	For any $\eps > 0$ there exists a
	constant $C(\eps) > 0$ such that for any $t\in\cI$ the following inequality
	\[
		\|u|_{\Sg_t}\|^2_{L^2(\Sg_t)} \le
		\eps\|\nabla u\|^2_{L^2(\dR^2;\dC^2)} + C(\eps)\|u\|^2_{L^2(\dR^2)}
	\]
	holds for all $u\in H^1(\dR^2)$. As a consequence, there is a constant $c > 0$ such that the inequality
	$\|u|_{\Sg_t}\|_{L^2(\Sg_t)}^2\le c\|u\|_{H^1(\dR^2)}^2$ holds for any $u\in H^1(\dR^2)$ and all $t\in\cI$.
\end{lem}
\begin{proof}
	In view of the density of $C^\infty_0(\dR^2)$ in $H^1(\dR^2)$ it suffices to check the inequality for $C^\infty_0$-functions.
	For any $u\in C^\infty_0(\dR^2)$,
	$s\in [0,L)$, and $t\in \cI$ we get
	using the fundamental theorem of calculus
	the following estimate
	\begin{equation}\label{eq:Dust}
	\begin{aligned}
	\cD_u(s,t) & := \left||u(\s(s)+t\nu(s))|^2 - |u(\s(s))|^2\right| \\
	&=
	\left|\int_0^1 \frac{\dd}{\dd r}
	(|u|^2(\s(s)+rt\nu(s))\,\dd r\right|\\
	&\le
	\int_0^1 |\langle\nabla (|u|^2)(\s(s)+rt\nu(s)),t\nu(s)\rangle|\,\dd r	\\
	& \le	
	2|t|
	\int_0^1 \big|(|u|\cdot\nabla |u|) (\s(s)+rt\nu(s))\big|\,\dd r.\\
	\end{aligned}
	\end{equation}		
	Let $\hat\eps > 0$ be arbitrary. 	
	Applying the inequality $\hat\eps a^2+\hat\eps^{-1}b^2 \ge
	2ab$ with $a,b > 0$, we can further estimate $\cD_u(s,t)$ as follows
	\begin{equation}\label{estimate0}
	\begin{aligned}
	\cD_u(s,t)&\le
	|t|
	\int_0^1
	\left(\hat\eps\big|\nabla |u| (\s(s)+rt\nu(s))\big|^2 + \hat\eps^{-1}
	\big|u (\s(s)+rt\nu(s))\big|^2\right) \dd r\\
	& \le
	\int_{-d_-}^{d_+}
	\left(\hat\eps\big|\nabla u (\s(s)+q\nu(s))\big|^2 + \hat\eps^{-1}
	\big|u (\s(s)+q\nu(s))\big|^2\right) \dd q,
	\end{aligned}
	\end{equation}
	where the substitution $q = rt$ was used, the interval of integration was enlarged and the
	diamagnetic inequality~\cite[Thm. 7.21]{LL} was applied.
	
	By the trace theorem~\cite[Thm. 3.38]{McL} (see also~\cite[Lem. 2.6]{BEL14}) for any $\tilde{\eps} > 0$ there exists $\tilde{C}(\tilde{\eps}) > 0$ such that
	\begin{equation}\label{estimate1}
		\|u|_\Sg\|^2_{L^2(\Sg)} =
		\int_0^L |u(\s(s))|^2\dd s \le \tilde{\eps}\|\nabla u\|^2_{L^2(\dR^2;\dC^2)} + \tilde{C}(\tilde{\eps})\|u\|^2_{L^2(\dR^2)}
	\end{equation}
	for any $u\in C^\infty_0(\dR^2)$.
	In view of~\eqref{eq:Jac} there exist constants $c_+ > c_- > 0$ such that
	\[
	1+\kp(s)t\in[c_-, c_+]\qquad\text{for all}\, s\in[0,L),\, t\in\cI.
	\]
	In this way we obtain the following simple estimate,
	\begin{equation}\label{estimate2}
	\|u|_{\Sg_t}\|^2_{L^2(\Sg_t)} =
	\int_0^L |u(\s(s)+t\nu(s))|^2(1+\kp(s)t)\,\dd s
	\le c_+
	\int_0^L |u(\s(s)+t\nu(s))|^2\,\dd s.
	\end{equation}
	Furthermore, combining estimates~\eqref{estimate0},~\eqref{estimate1}, and~\eqref{estimate2}, we end up with
	\[		
	\begin{aligned}
	\|u|_{\Sg_t}\|^2_{L^2(\Sg_t)} &\le
	c_+
	\int_0^L \left(|u(\s(s))|^2+ \cD_u(s,t)
	\right) \dd s\\
	&\le
	c_+\tilde\eps\|\nabla u\|^2_{L^2(\dR^2;\dC^2)} + c_+ \tilde{C}(\tilde\eps)\|u\|^2_{L^2(\dR^2)}\\
	&\qquad+ \frac{\hat\eps c_+}{c_-}\int_0^L\int_{-d_-}^{d_+}
	\big|\nabla u (\s(s)+q\nu(s))\big|^2(1+\kp(s)q)\,\dd q\dd s\\
	&\qquad\qquad +
	\frac{c_+}{\hat\eps c_-}
	\int_0^L\int_{-d_-}^{d_+}
	\big|u (\s(s)+q\nu(s))\big|^2(1+\kp(s)q)\,\dd q\dd s\\
	&\le
	\left(c_+\tilde\eps + \frac{c_+\hat\eps}{c_-}\right)
	\|\nabla u\|^2_{L^2(\dR^2;\dC^2)}
	+ \left(c_+ \tilde{C}(\tilde\eps) + \frac{c_+}{c_-\hat\eps}\right)\|u\|^2_{L^2(\dR^2)}.
	\end{aligned}	
	\]
	By choosing $\tilde\eps >0$ and $\hat\eps>0$ such that $\eps =c_+\tilde\eps + \frac{c_+\hat\eps}{c_-}$ we get the sought claim.
\end{proof}	
In the next lemma we prove that the function $\cI\ni t\mapsto \|u|_{\Sg_t}\|^2_{L^2(\Sg_t)}$ is H\"older continuous with exponent $\frac12$ for any $u\in H^1(\dR^2)$.
\begin{lem}\label{lem:trace_cont}
	There exists a constant $c >0$ such that
	\[
		\left|\|u|_{\Sg_{t_2}}\|_{L^2(\Sg_{t_2})}^2 - \|u|_{\Sg_{t_1}}\|^2_{L^2(\Sg_{t_1})}\right| \le c|t_1 - t_2|^{1/2}\|u\|^2_{H^1(\dR^2)}
	\]
	holds for all $u\in H^1(\dR^2)$ and for any $t_1,t_2\in \cI$.
\end{lem}
\begin{proof}
		Throughout the proof $c > 0$ denotes a generic positive constant, which
		varies from line to line.
		In view of the density of $C^\infty_0(\dR^2)$ in $H^1(\dR^2)$ it suffices to check the inequality for $C^\infty_0$-functions.
		Without loss of generality we may prove the claim only for the case that $t_1 = 0$ and that $t_2 = t \in (0,d_+)$.
		In this case we need to show that
		\[
			\cS_u(t) := \left|\|u|_{\Sg_{t}}\|_{L^2(\Sg_{t})}^2 - \|u|_{\Sg}\|^2_{L^2(\Sg)}\right| \le ct^{1/2}\|u\|^2_{H^1(\dR^2)}.
		\]
		By elementary means we obtain the bound
		\[	
		\begin{aligned}
			\cS_u(t) &= \left|
			\int_0^L|u(\s(s)+t\nu(s))|^2(1+\kp(s)t)\,\dd s
			-
			\int_0^L|u(\s(s))|^2\,\dd s
			\right|\\
			& \le ct\int_0^L|u(\s(s)+t\nu(s))|^2\,\dd s +
			\int_0^L\left||u(\s(s)) + t\nu(s))|^2 - |u(\s(s))|^2\right|\dd s
		\end{aligned}
		\]
		Using Lemma~\ref{lem:trace} and the estimate of~\eqref{eq:Dust}, and taking~\eqref{eq:Jac} into account we find
		\begin{equation}\label{eq:Sbnd}
	\begin{aligned}
			\cS_u(t)
			&\le
 			ct
			\|u|_{\Sg_t}\|^2_{L^2(\Sg_t)}
			+2\int_0^L\int_0^t|(|u|\cdot\nabla|u|)(\s(s)+q\nu(s))|\,\dd q\,\dd s\\
			& \le
			ct^{1/2}\|u\|^2_{H^1(\dR^2)} +
			2\int_0^L\int_0^t|(|u|\cdot\nabla|u|)(\s(s)+q\nu(s))|\,\dd q\,\dd s.
		\end{aligned}
		\end{equation}
		Furthermore, applying Cauchy-Schwarz inequality, diamagnetic inequality, and Lemma~\ref{lem:trace} again, we get
		\[
	\begin{aligned}
			&
			\int_0^L\int_0^t|(|u|\cdot\nabla|u|)(\s(s)+q\nu(s))|\,\dd q\,\dd s\\
			&\qquad \le
			\left(\int_0^L\int_0^t|\nabla u(\s(s)+q\nu(s))|^2\,\dd q\dd s\right)^{1/2}
			\left(\int_0^L\int_0^t|u(\s(s)+q\nu(s))|^2\,\dd q\dd s\right)^{1/2}\\
			&\qquad\le
			c\left(\int_0^L\int_0^t|\nabla u(\s(s)+q\nu(s))|^2(1+\kp(s)q)\,\dd q\dd s\right)^{1/2}
						\left(\int_0^t\|u|_{\Sg_q}\|_{L^2(\Sg_q)}^2\,\dd q\right)^{1/2}
		\\
		&\qquad \le
		ct^{1/2}\|\nabla u\|_{L^2(\dR^2;\dC^2)}\|u\|_{H^1(\dR^2)}
		\le ct^{1/2}\|u\|^2_{H^1(\dR^2)}.				
	\end{aligned}
		\]
	Combining the last estimate with~\eqref{eq:Sbnd} we arrive at the claim.	
\end{proof}
The purpose of the last lemma of the appendix is to establish the continuity of the lowest eigenvalue of $\sfH_{\aa\dl_{\Sg_t}}$ with respect to $t$.
\begin{lem}\label{lem:Op_cont}
	Let the curves $\Sg_t$ be as in~\eqref{eq:Sgt} and the number $\aa > 0$ be fixed.
	The operator-valued function
	\begin{equation}\label{eq:operator_cont}
		\cI\ni t\mapsto \sfH_{\aa\dl_{\Sg_t}}
	\end{equation}
	is continuous in the norm-resolvent topology and uniformly lower-semibounded, and as a consequence, the function $\cI\ni t\mapsto \lm_1(\aa\dl_{\Sg_t})$ is continuous.
\end{lem}
\begin{proof}
	Throughout the proof $c > 0$ again denotes a generic positive constant, which
	varies from line to line.
	Lemma~\ref{lem:trace} in combination with the expression
	for the form~\eqref{eq:form} referring to $\mu = \aa\dl_{\Sg_t}$ shows that operators $\sfH_{\aa\dl_{\Sg_t}}$, $t\in\cI$, are uniformly bounded from below by some constant $\lm_1< 0$.
	Without loss of generality it suffices to prove that $\sfH_{\aa\dl_{\Sg_t}}$ converges in the norm resolvent sense to $\sfH_{\aa\dl_\Sg}$ as $t\arr 0$. The norm resolvent continuity of $\cI\ni t\mapsto \sfH_{\aa\dl_{\Sg_t}}$  at the other points of the interval $\cI$ can be proven analogously.
	
	We fix $\lm_0 < \lm_1$, use the notation $\sfR_t := (\sfH_{\aa\dl_{\Sg_t}}-\lm_0)^{-1}$, $t\in\cI$ and claim
	that there is a constant $c > 0$ such that
	\begin{equation}\label{eq:resolvent_estimate}
		\|\sfR_t - \sfR_0\| \le c|t|^{1/2}.
	\end{equation}
	In fact, first we note that
	\[
	\begin{aligned}
		\|\sfR_t - \sfR_0\|
		& =
		\sup_{\|u\|, \|v\| = 1}
		\big|((\sfR_t - \sfR_0)u,v)_{L^2(\dR^2)}\big|\\
		& =
		\sup_{\|u\|, \|v\| = 1}
		\big|(\sfR_t u,(\sfH_{\aa\dl_\Sg} - \lm_0) \sfR_0v)_{L^2(\dR^2)}-
		((\sfH_{\aa\dl_{\Sg_t}} - \lm_0)\sfR_tu, \sfR_0v)_{L^2(\dR^2)}
		\big|\\
		& =
		\sup_{\|u\|, \|v\| = 1}
		\big|
		\frh_{\aa\dl_{\Sg}}[\sfR_t u,\sfR_0v]-
		\frh_{\aa\dl_{\Sg_t}}[\sfR_t u,\sfR_0v]
		\big|.
	\end{aligned}
	\]
	The estimate~\eqref{eq:resolvent_estimate} would follow if we prove that
	\begin{equation}\label{eq:fgineq}
		\big|
		\frh_{\aa\dl_{\Sg}}[f,g]-
		\frh_{\aa\dl_{\Sg_t}}[f,g]
		\big|\le c|t|^{1/2}\left(\|f\|^2_{H^1(\dR^2)}
		+\|g\|^2_{H^1(\dR^2)} \right),\quad f,g\in H^1(\dR^2),
	\end{equation}
	since with the choice $f = \sfR_t u$ and $g = \sfR_0v$ the inequality~\eqref{eq:fgineq} together with Lemma~\ref{lem:trace} yields
	the existence of constants $c_1 >0$ and $c_2 >-c_1\lm_0$ such that
	\[
	\begin{aligned}
		&\big|
		\frh_{\aa\dl_{\Sg}}[\sfR_t u,\sfR_0v]-
		\frh_{\aa\dl_{\Sg_t}}[\sfR_t u,\sfR_0v]
		\big|\\		
		&\qquad \le c|t|^{1/2}
		\left(
		c_1\frh_{\aa\dl_{\Sg}}[\sfR_0 v] + c_2\|\sfR_0 v\|^2_{L^2(\dR^2)} + c_1\frh_{\aa\dl_{\Sg_t}}[\sfR_t u] + c_2\|\sfR_t u\|^2_{L^2(\dR^2)}
		\right)\\
		&\qquad= c|t|^{1/2}\left[
		c_1(\sfR_0v,v)_{L^2(\dR^2)} + c_1(\sfR_tu,u)_{L^2(\dR^2)} +
		(c_1\lm_0 + c_2)\left(\|\sfR_0v\|^2_{L^2(\dR^2)}
		+\|\sfR_tu\|^2_{L^2(\dR^2)}\right)
		\right]\\
		&\qquad \le
		c|t|^{1/2}\left(\|u\|^2_{L^2(\dR^2)} + \|v\|^2_{L^2(\dR^2)}\right),
	\end{aligned}
	\]
	where we used that $\|\sfR_t\|\le \frac{1}{\lm_1-\lm_0}$ holds for all $t\in\cI$.
	In view of the polarization
	identity it is enough to check~\eqref{eq:fgineq} for $f = g$. By definition of the form $\frh_{\aa\dl_{\Sg_t}}$ in~\eqref{eq:form} with $\mu =\aa\dl_{\Sg_t}$ we get
	\[
	\begin{aligned}
		\big|\frh_{\aa\dl_{\Sg_t}}[f] - \frh_{\aa\dl_\Sg}[f]\big| =
		\aa\left|\|f|_{\Sg_t}\|^2_{L^2(\Sg_t)}-
		\|f|_{\Sg}\|^2_{L^2(\Sg)}\right| \le c|t|^{1/2}\|f\|_{H^1(\dR^2)}^2,
	\end{aligned}
	\]
	where Lemma~\ref{lem:trace_cont} was applied in the last step. The continuity of the lowest eigenvalue
	$\cI\ni t\mapsto \lm_1(\aa\dl_{\Sg_t})$ follows from the norm resolvent continuity of the operator-valued function in~\eqref{eq:operator_cont} in combination with the spectral convergence result from~\cite[Satz 9.24]{W}.
\end{proof}
\end{appendix}

\section*{Acknowledgement}
The work of P.E. was supported by the EU project CZ.02.1.01/0.0/0.0/16\textunderscore 019/0000778.

\newcommand{\etalchar}[1]{$^{#1}$}

\end{document}

%% file: commands.tex
\newcommand\nb{\nabla}
\newcommand{\beq}{\begin{equation} \begin{split}}
\newcommand{\eeq}{\end{split} \end{equation}}
\newcommand\Sg{\Sigma}
\newcommand\Omg{\Omega}

\makeatletter
\def\section{\@startsection{section}{1}\z@{.9\linespacing\@plus\linespacing}%
	{.7\linespacing} {\fontsize{13}{14}\selectfont\bfseries\centering}}
\def\paragraph{\@startsection{paragraph}{4}%
	\z@{0.3em}{-.5em}%
	{$\bullet$ \ \normalfont\itshape}}

\@addtoreset{equation}{section}
\makeatother

\renewcommand\and{\qquad\text{and}\qquad}

\newcommand\sm{\setminus}
\newcommand\Op{\sfH_{\mu}}

\newcommand\dl{\delta}
\newcommand\one{\mathbbm{1}}

\newcommand{\comm}[1]{}

\def\sfH{\mathsf{H}}

\def\bm1{\mathbbm{1}}

\def\s{\sigma}

\def\p{\partial}

\def\one{\mathbbm{1}}

\def\arr{\rightarrow}

\def\tt{\theta}
\def\aa{\alpha}
\def\lm{\lambda}

\def\s{\sigma}

\def\ii{{\mathsf{i}}}
\def\p{\partial}

\def\kp{\kappa}

\def\sfH{\mathsf{H}}

\def\one{\mathbbm{1}}

\def\dd{{\mathsf{d}}}

\newcounter{counter_a}

\usepackage[latin1]{inputenc}
\usepackage[T1]{fontenc}

\newcommand{\eg}{{\it e.g.}\,}

\newcommand{\cf}{{\it cf.}\,}

\numberwithin{figure}{section}
\numberwithin{equation}{section}
\theoremstyle{plain}
\newtheorem*{thm*}{Theorem}
\newtheorem{thm}{Theorem}[section]
\newtheorem{hyp}{Hypothesis}[section]
\newtheorem{lem}[thm]{Lemma}
\newtheorem{prop}[thm]{Proposition}

\newtheorem{dfn}[thm]{Definition}
\theoremstyle{remark}
\newtheorem{remark}[thm]{Remark}
\theoremstyle{plain}

\newcommand{\beu}{\begin{equation*}}
\newcommand{\eeu}{\end{equation*}}
\newcommand{\besu}{\begin{equation*}
\begin{aligned}}
\newcommand{\eesu}{\end{aligned}
\end{equation*}}
\newcommand{\bes}{\begin{equation}
\begin{aligned}}
\newcommand{\ees}{\end{aligned}
\end{equation}}

\newcommand\cA{\mathcal A}

\newcommand\cD{\mathcal D}

\newcommand\cH{\mathcal H}

\newcommand\cS{\mathcal S}

\newcommand\frh{\mathfrak h}

\newcommand\eps{\varepsilon}

\newcommand\ov{\overline}

\newcommand\void[1]{}

\def\ov{\overline}
\def\eps{\varepsilon}
\def\ran{{\rm ran\,}}

      \def\dC{{\mathbb C}}

   \def\dN{{\mathbb N}}   
      \def\dR{{\mathbb R}}

   \def\dZ{{\mathbb Z}}

   \def\sfH{{\mathsf H}}

      \def\sfR{{\mathsf R}}

\def\cA{{\mathcal A}}      \def\cC{{\mathcal C}}
\def\cD{{\mathcal D}}      
   \def\cH{{\mathcal H}}   \def\cI{{\mathcal I}}

\def\cS{{\mathcal S}}      \def\cU{{\mathcal U}}

\newcommand{\dom}{\mathrm{dom}\,}